\documentclass[11pt,letterpaper]{article}

\usepackage{ramanujan}

\newcommand{\fix}{\mathrm{fix}}
\newcommand{\Cyc}{\mathrm{Cyc}}

\begin{document}

\title{Spectrum preserving short cycle removal on regular graphs} \author{Pedro Paredes\thanks{Computer Science Department, Carnegie Mellon University.  Supported by NSF grant CCF-1717606. This material is based upon work supported by the National Science Foundation under grant numbers listed above. Any opinions, findings and conclusions or recommendations expressed in this material are those of the author and do not necessarily reflect the views of the National Science Foundation (NSF).}} \date{\today}
\maketitle

\begin{abstract}
  We describe a new method to remove short cycles on regular graphs
  while maintaining spectral bounds (the nontrivial eigenvalues of the
  adjacency matrix), as long as the graphs have certain combinatorial
  properties. These combinatorial properties are related to the number
  and distance between short cycles and are known to happen with high
  probability in uniformly random regular graphs.

  Using this method we can show two results involving high girth
  spectral expander graphs. First, we show that given $d \geq 3$ and
  $n$, there exists an explicit distribution of $d$-regular
  $\bigTheta{n}$-vertex graphs where with high probability its samples
  have girth $\Omega(\log_{d - 1} n)$ and are
  $\epsilon$-near-Ramanujan; i.e., its eigenvalues are bounded in
  magnitude by $2\sqrt{d - 1} + \epsilon$ (excluding the single
  trivial eigenvalue of $d$). Then, for every constant $d \geq 3$ and
  $\epsilon > 0$, we give a deterministic $\poly(n)$-time algorithm
  that outputs a $d$-regular graph on $\bigTheta{n}$-vertices that is
  $\epsilon$-near-Ramanujan and has girth $\Omega(\sqrt{\log n})$,
  based on the work of \cite{MOP19b}.
\end{abstract}

\clearpage

\section{Introduction}

Let's consider $d$-regular graphs of $n$ vertices. The study of short
cycles and \textit{girth} (defined as the length of the shortest cycle
of a graph) in such graphs dates back to at least the 1963 paper of
Erd\H{o}s and Sachs \cite{ES63}, who showed that there exists an
infinite family with girth at least $(1 - o(1)) \log_{d - 1} n$. On
the converse side, a simple path counting argument known as the
``Moore bound'' shows that this girth is upper bounded by
$(1 + o_n(1)) 2 \log_{d - 1} n$. Though simple, this is the best known
upper bound. Given these bounds, it is common to call an infinite
family of $d$-regular $n$-vertex graphs \textit{high girth} if their
girth is $\Omega(\log_{d - 1} n)$.

The first explicit construction of high girth regular graphs is
attributed to Margulis \cite{Mar82}, who gave a construction of graphs
that achieve girth $(1 - o(1)) \frac{4}{9} \log_{d - 1} n$. A series
of works initiated by Lubotzky-Phillips-Sarnak \cite{LPS88} and then
improved by several other people \cite{Mar88,Mor94,LU95} culminated in
the work of Dahan \cite{Dah14}, who proves that for all large enough
$d$ there are explicit $d$-regular $n$-vertex graphs of girth
$(1 - o(1)) \frac{4}{3} \log_{d - 1} n$.

Another relevant problem consists of generating random distributions
that produce regular graphs with high girth. Results regarding the
probabilistic aspects of certain certain structures (like cycles) in
graphs often give us tools to count the number of graphs that satisfy
certain conditions, like how many regular graphs have girth at least
some value. The distribution of short cycles in uniformly random
regular graphs was first studied by Bollob\'{a}s \cite{Bol80}, who
proved, that for a fixed $k$ the random variables representing the
number of cycles of length at most $k$ in a uniformly random
$d$-regular graph are asymptotically independent Poisson with mean
$(d - 1)^i / 2i$, where $i$ is the length of the cycle. Subsequently,
McKay-Wormald-Wysocka \cite{MWW04} gave a more precise description of
this by finding the asymptotic probability of a random $d$-regular
graph having a certain number of cycles of any length up to
$c \log_{d - 1} n$, for $c < 1/2$. More recently, Linial and Simkin
\cite{LS19} showed that a random greedy algorithm that is given
$d \geq 3$, $c \in (0, 1)$ and an even $n$, produces a $d$-regular
$n$-vertex graph with girth at least $c \log_{d - 1} n$ with high
probability.

The literature of regular graphs with high girth is closely connected
to the literature of \textit{spectral expanders}. Before defining
this, let's consider some notation.

\begin{definition}
  Let $G$ be an $n$-vertex $d$-regular multigraph. We write
  $\lambda_i = \lambda_i(G)$ for the eigenvalues of its adjacency
  matrix $A_G$, and we always assume they are ordered with
  $\lambda_1 \geq \lambda_2 \geq \cdots \geq \lambda_n$. A basic fact
  is that $\lambda_1= d$ always; this is called the trivial eigenvalue
  and corresponds to the all ones vector. We also write
  $\lambda(G) = \max\{\lambda_2, |\lambda_n|\}$.
\end{definition}

Roughly, a graph with good spectral expansion properties is a graph
that has small $\lambda$. More formally, an infinite sequence $(G_n)$
of $d$-regular graphs is called a family of expanders if there is a
constant $\delta > 0$ such that $\lambda(G) \leq (1 - \delta) d$ for
all $n$, or in other words, all eigenvalues are strictly separated
from the trivial eigenvalue. This terminology was first introduced by
\cite{Pin73} and later it was shown \cite{Alo86} that uniformly random
$d$-regular graphs are spectral expanders with high probability.

The celebrated Alon-Boppana bound shows that $\lambda$ cannot be
arbitrarily small:

\begin{theorem}(\cite{Alo86,Nil91,Fri93}). For any $d$-regular $n$-vertex graph
  $G$ we have that
  $\lambda_2(G) \geq 2\sqrt{d - 1} - O(1 / \log^2 n)$.
\end{theorem}

Using some number-theoretic ideas, Lubotzky-Phillips-Sarnak
\cite{LPS88}, and independently Margulis \cite{Mar88}, proved this
bound is essentially tight by showing the existence of infinite
families of $d$-regular graphs that meet the bound
$\lambda(G) \leq 2\sqrt{d - 1}$, if $d - 1$ is an odd prime. In light
of this, Lubotzky-Phillips-Sarnak introduced the following definition:

\begin{definition}(Ramanujan graphs). A $d$-regular graph $G$ is
  called Ramanujan whenever $\lambda(G) \leq 2\sqrt{d - 1}$.
\end{definition}

These results were improved by Morgenstern \cite{Mor94}, who showed
the same for all $d$ where $d - 1$ is a prime power.

It is still open whether there exist infinite families of Ramanujan
graphs for all $d$. However, if one relaxes this to only seek
$\epsilon$-near-Ramanujan graphs (graphs that satisfy
$\lambda \leq 2\sqrt{d - 1} + \epsilon$), then the answer is
positive. Friedman \cite{Fri08} proved that uniformly random
$d$-regular $n$-vertex graphs satisfy
$\lambda \leq 2\sqrt{d - 1} + o_n(1)$ with high probability. This
proof was recently simplified by Bordenave \cite{Bor19}.

\begin{theorem}(\cite{Fri08,Bor19}). \label{thm:bor}
  Fix any $d > 3$ and $\epsilon > 0$ and let $G$ be a uniformly random
  $d$-regular $n$-vertex graph. Then

  \[
    \prob{\lambda(G) \leq 2\sqrt{d - 1} + \epsilon} \geq 1 - o_n(1).
  \]

  In fact \cite{Bor19}, $G$ achieves the subconstant
  $\epsilon = \bigOTilda{1 / \log^2 n}$ with probability at least
  $1 - 1/n^{.99}$.
\end{theorem}

Recently, it was shown how to achieve a result like the above but
deterministically \cite{MOP19b}. We write a more precise statement of
this below.

\begin{theorem}(\cite{MOP19b}).
  \label{thm:mop}
  Given any $n$, $d \geq 3$ and $\epsilon > 0$, there is deterministic
  polynomial-time algorithm that constructs a $d$-regular
  $N$-vertex graphs with the following properties:

  \begin{itemize}
  \item $N = n(1 + o_n(1))$;
  \item $\lambda(G) \leq 2\sqrt{d - 1} + \epsilon$;
  \end{itemize}
\end{theorem}

We refer the reader interested in a more thorough history of the
literature of Ramanujan graphs to the introduction of
\cite{MOP19b}. Also, for a comprehensive list of applications and
connections of Ramanujan graphs and expanders to computer science and
mathematics, see \cite{HLW06}.

In this work we concern ourselves with bridging these two worlds,
looking for families of regular graphs that are both good spectral
expanders and also have high girth. This bridge can be seen in several
of the aforementioned works. The explicit construction of high girth
regular graphs by Margulis \cite{Mar82} was a motivator to his work on
Ramanujan graphs \cite{Mar88}. Additionally, the constructions of
\cite{LPS88} and \cite{Mor94} produce graphs that are both Ramanujan
and have girth $(1 - o(1)) \frac{4}{3} \log_{d - 1} n$, according to
the previously stated restrictions on $d$.

More recently, Alon-Ganguly-Srivastava \cite{AGS19} showed that for a
given $d$ such that $d - 1$ is prime and $\alpha \in (0, 1/6)$, there
is a construction of infinite families of graphs with girth at least
$(1 - o(1)) (2/3) \alpha \log_{d - 1} n$ and $\lambda$ at most
$(3 / \sqrt{2}) \sqrt{d - 1}$ with many eigenvalues localized on small
sets of size $O(n^\alpha)$. Our main result is based on the techniques
of this work.

One motivation to search for graphs with simultaneous good spectral
expansion and high girth is its application to the theory of
error-correcting codes, particularly for \textit{Low Density Parity
  Check} or \textit{LDPC} codes, originally introduced by Gallager
\cite{Gal62}. The connection with high girth regular graphs was first
pointed out by Margulis in \cite{Mar82}. The property of high-girth is
desirable since the decoding of such codes relies on an iterative
algorithm whose performance is worse in the presence of short
cycles. Additionally, using graphs with good spectral properties to
generate these codes seems to lead to good performance, as pointed out
by several works \cite{RV00,LR00,MS02}.

\subsection{Our results}

We can now state our results and put them in perspective. Let's first
introduce some useful definitions and notation.

\begin{definition}[Bicycle-free at radius $r$]
  A multigraph is said to be \textit{bicycle-free at radius} $r$ if
  the distance-$r$ neighborhood of every vertex has at most one cycle.
\end{definition}

\begin{definition}[$(r, \Lambda, \tau)$-graph]
  Let $r$ and $\tau$ be a positive integers and $\Lambda$ be a
  positive real. Then, we call a graph $G$ a
  $(r, \Lambda, \tau)$-graph if it satisfies the following conditions:
  
  \begin{itemize}
  \item $G$ is bicycle-free at radius at least $r$;
  \item $\lambda(G) \leq \Lambda$;
  \item The number of cycles of length at most $r$ is at most $\tau$.
  \end{itemize}
\end{definition}

Our main result is the following short cycle removal theorem:

\begin{theorem}
  \label{thm:cycrem}
  There exists a deterministic polynomial-time algorithm $\fix$ that,
  given as input a $d$-regular $n$-vertex $(r, \Lambda, \tau)$-graph
  $G$ satisfying
  
  \begin{equation*}
    \label{eq:cycremcons}
    \Lambda \geq 2\sqrt{d - 1}, \qquad r \leq \frac{2}{3}\log_{d - 1}(n / \tau) - 5,
  \end{equation*}
  
  \noindent outputs a graph $\fix(G)$ satisfying

  \begin{itemize}
  \item $\fix(G)$ is a $d$-regular graph with
    $n + O(\tau \cdot (d-1)^{r / 2 + 1})$ vertices;
  \item $\lambda(\fix(G)) \leq \Lambda + O(1 / r)$;
  \item $\fix(G)$ has girth at least $r$.
  \end{itemize}
\end{theorem}

Our proof of this statement uses several ideas from \cite{AGS19}. We
will prove this theorem in \Cref{sec:rem}.

The preconditions of this theorem are not arbitrary. Even though
random uniformly $n$-vertex $d$-regular graphs have constant girth
with high probability, they are bicycle-free at radius
$\Omega(\log_{d-1} n)$ and the number of cycles of length at most
$c \log_{d-1} n$ (for small enough $c$) is $o(n)$ with high
probability. Recall that from \Cref{thm:bor} we also know that being
near-Ramanujan is also a property that occurs with high probability in
random regular graphs. So a statement like the above can be used to
produce distributions over regular graphs that have high girth and are
near-Ramanujan with high probability. With this in mind, we introduce
the following definition:

\begin{definition}(($\Lambda$, $g$)-good graphs).  We call a graph $G$
  a ($\Lambda$, $g$)-good graph if $\lambda(G) \leq \Lambda$ and
  $\text{girth}(G) \geq g$.
  
  Let $\mu_d(n)$ be a distribution over $d$-regular graphs with
  $\sim n$ vertices. We say $\mu_d(n)$ is ($\Lambda$, $g$)-good if
  $G \sim \mu_d(n)$ is ($\Lambda$, $g$)-good with probability at least
  $1 - o_n(1)$.

  Additionally, we call the distribution explicit if sampling an
  element is doable in polynomial time.
\end{definition}

We shall prove the following using \Cref{thm:cycrem} in
\Cref{sec:apprand}:

\begin{theorem} \label{thm:apprand} Given $d \geq 3$ and $n$, let $G$
  be a uniformly random $d$-regular $n$-vertex graph. For any
  $c < 1 / 4$ and $\epsilon > 0$, $\fix(G)$ is a
  $(2\sqrt{d - 1} + \epsilon, c\log_{d - 1} n)$-good explicit
  distribution.
\end{theorem}

Recall that the upper bound on the girth of a regular graph is
$(1 + o_n(1)) 2 \log_{d - 1} n$, so this distribution has optimal
girth up to a constant. Based on our proof of the above and using
some classic results about the number of $d$-regular $n$-vertex
graphs, we can show a lower bound on the number of
($2\sqrt{d - 1} + \epsilon$, $c \log_{d - 1} n$)-good graphs in some
range.

\begin{corollary} \label{cor:apprand} Let $d \geq 3, n$ be integers
  and $\epsilon > 0, c > 1/4$ reals. The number of $d$-regular graphs
  with number of vertices in $[n, n + O(n^{3/8})]$, which are
  ($2\sqrt{d - 1} + \epsilon$, $c \log_{d - 1} n$)-good, is at least

  \[
    \Omega \parens*{\parens*{\frac{d^dn^d}{e^d(d!)^2}}^{n / 2}}.
  \]
\end{corollary}

We prove both of these results in \Cref{sec:apprand}.

Finally, we show a slightly stronger version of result of
\cite{MOP19b} by plugging our short cycle removal theorem into their
construction.

\begin{theorem}
  \label{thm:appdet}
  Given any $n$, $d \geq 3$, $\epsilon > 0$ and $c$ such that

  \[
    c \leq \frac{\sqrt{\log n} \cdot \log_{d - 1} 2}{15},
  \]

  \noindent there is deterministic polynomial-time (in $n$) algorithm
  that constructs a $d$-regular $N$-vertex graphs with the following
  properties:

  \begin{itemize}
  \item $N = n(1 + o_n(1))$;
  \item $\lambda(G) \leq 2\sqrt{d - 1} + \epsilon$;
  \item $G$ has girth at least $c\sqrt{\log n}$.
  \end{itemize}
\end{theorem}

We prove this result in \Cref{sec:appdet}.

\subsection{Models of random regular graphs}

We will introduce some classic models of random regular graphs, which
we will use throughout the paper.

\begin{definition}[$\calG_d(n)$]
  Let $\calG_d(n)$ denote the set of $d$-regular $n$-vertex graphs. We
  write $G \sim \calG_d(n)$ to denote that $G$ is sampled uniformly at
  random from $\calG_d(n)$.
\end{definition}

Sampling from $\calG_d(n)$ is not easy a priori; the standard way to
do so is using the \textit{configuration model}, which was originally
defined by Bollob\'{a}s \cite{Bol80}.

\begin{definition}(Configuration model).  Given integers $n > d > 0$
  with $nd$ even, the configuration model produces a random
  $n$-vertex, $d$-regular undirected multigraph (with loops) $G$. This
  multigraph is induced by a uniformly random matching on the set of
  ``half-edges'', $[n] \times [d] \cong [nd]$ (where
  $(v,i) \in [n] \times [d]$ is thought of as half of the $i$th edge
  emanating from vertex $v$).  Given a matching, the multigraph $G$ is
  formed by ``attaching'' the matched half-edges.
\end{definition}

This model corresponds exactly to the uniform distribution on not
necessarily simple $d$-regular $n$-vertex graphs. It also not hard to
see that the conditional distribution of the $d$-regular $n$-vertex
configuration model when conditioned on it being a simple graph is
exactly the uniform distribution on $\calG_d(n)$. The probability that
the sampled graph is simple is $\Omega_d(1)$.

The configuration model has the advantage that is easy to sample and
to analyze. For reference, the proof of \Cref{thm:bor} was done in
terms of the configuration model and so the theorem also applies to
it.

\section{Short cycles removal} \label{sec:rem}

In this section we prove \Cref{thm:cycrem}. Recall that we are given a
$d$-regular $n$-vertex $(r, \Lambda, \tau)$-graph $G$ with the
constraints specified in \Cref{eq:cycremcons} and we wish to find some
$d$-regular graph $\fix(G)$ on $\sim n$ vertices such that
$\lambda(\fix(G)) \leq \Lambda + o_r(1)$ and its girth is at least
$r$.

Briefly, the algorithm that achieves this works by removing one edge
per small cycle from $G$, effectively breaking apart all such cycles,
and then fixing the resulting off degree vertices by adding $d$-ary
trees in a certain way. We will now more carefully outline this method
and then proceed to fill in some details as well as show it works as
desired.

Before starting, we introduce some notation which will be helpful.

\begin{definition}[$\Cyc_g(G)$]
  Given a graph $G$, let $\Cyc_g(G)$ denote the collection of all
  cycles in $G$ of length at most $g$. Recall that if $\Cyc_g(G)$ is
  empty then $G$ is said to have girth exceeding $g$.
\end{definition}

\begin{definition}[$B_\delta(S)$]
  Given a set of vertices $S$ in a graph $G$, let $B_\delta(S)$ denote
  the collection of vertices in $G$ within distance $\delta$ of
  $S$. We will occasionally abuse this notation and write
  $B_\delta(v)$ instead of $B_\delta(\{v\})$ for a vertex $v$.
\end{definition}

Let $E_c$ be a set containing exactly one arbitrary edge per cycle in
$\Cyc_{r}(G)$ and let $H_t$ be a graph with the same vertex set as $G$
obtained by removing all edges in $E_c$ from $G$. To prevent
ambiguity, whenever we pick something arbitrarily let's suppose the
algorithm $\fix$ uses the lexicographical order of node labels as a
tiebreaker. We also partition the endpoints of each edge as described
in the following definition:

\begin{definition}
  Given an edge set $E$, we let $V_1(E)$ and $V_2(E)$ be two disjoint
  sets of vertices constructed as follows: for all $e = (u, v) \in E$
  place $u$ in $V_1(E)$ and $v$ in $V_2(E)$ (so each endpoint is in
  exactly one of the two sets).
\end{definition}

Note that according to the above definition we have
$|V_1(E_c)| = |V_2(E_c)| = |E_c| \leq \tau$. For ease of notation we
also define:

\begin{definition}[$\phi_E(v)$]
  Given an edge set $E$ and $(u, v) \in E$ such that $u \in V_1(E)$
  and $v \in V_2(E)$, we denote by $\phi_E$ the function that maps
  endpoints to endpoints, so we have $\phi_E(u) = v$ and
  $\phi_E(v) = u$.
\end{definition}

We will often abuse notation and drop the $E$ from $\phi_E$ when it is
clear from context.

Since we break apart each cycle in $\Cyc_{r}(G)$, we can conclude that
$H_t$ has girth more than $r$. However, note that in removing edges,
$H_t$ is no longer $d$-regular.

To fix this, consider the following object which we refer to as a
$d$-regular tree of height $h$: a finite rooted tree of height $h$
where the root has $d$ children but all other non-leaf vertices have
$d-1$ children. This definition has that every non-leaf vertex in
a $d$-regular tree has degree $d$.

We shall add two $d$-regular trees to $H_t$ in order to fix the off
degrees, while maintaining the desired girth and bound on
$\lambda$. The idea of using $d$-regular trees is based on the
degree-correction gadget used in \cite{AGS19} for their construction
of high-girth near-Ramanujan graphs with localized eigenvectors. As
such, we will use some of the tools used in their proofs.

Let $h$ be an integer parameter we shall fix later. Let $T_1$ and
$T_2$ be two $d$-regular trees of height $h$ and let $L_1$ and $L_2$
be the sets of leaves of each one. Note that
$|L_1| = |L_2| = d(d - 1)^{h - 1} \approx (d - 1)^h$. We shall add
the two trees to $H_t$ and then pair up elements of $V_1(E_c)$ with
elements of $L_1$ (and analogously for $V_2(E_c)$ and $L_2$) and merge
the paired up vertices. However, we have to deal with two potential
issues:

\begin{itemize}
\item $|L_i| \neq |V_i(E_c)|$, in which case we cannot get an exact
  pairing between these sets;
\item This procedure might result in the creation of small cycles
  (potentially even cycles of length $O(1)$).
\end{itemize}

To expand on the latter point, we describe a potential problematic
instance. Suppose we can somehow pick $h$ such that
$|L_i| = |V_i(E_c)|$ and then arbitrarily pair up their
elements. Suppose there are two edges in $E_C$ corresponding to two
cycles of constant length and denote their endpoints by
$v_1 \in V_1(E_C), v_2 \in V_2(E_C)$ and
$u_1 \in V_1(E_C), u_2 \in V_2(E_C)$. If the distance in $T_1$ of
$v_1$ and $u_1$ given by the pairing of $V_1(E_c)$ and $L_1$ is small
(constant, for example) and the same applies to the distance in $T_2$
of $v_2$ and $u_2$, then there is a cycle of small length (constant,
for example) in the graph resulting from adding the two trees to
$H_t$.

To address this issue we remove some extra edges from $G$ that are
somehow ``isolated'' and group them with edges from $E_C$. The goal is
to have the endpoints of any two edges in $E_C$ be far apart in $T_1$
and $T_2$ distance, but close to some of the endpoints of the extra
edges. With this in mind, we set
$h = \ceil{\log_{d-1} \tau} + \ceil{r / 2} + 1$ so that
$|L_i| \approx \tau \cdot (d - 1)^{r / 2 + 1}$, which is close to the
number of extra edges we want to remove. This choice will also be
helpful later when we analyze the spectral properties of the
construction.

Formally, this leads us to the following proposition:

\begin{proposition}
  \label{prop:treeadd}
  There is a set of edges $E_t$ of $G$ such that the following is true
  for $i \in \{1, 2\}$:

  \begin{itemize}
  \item $|V_i(E_t) \cup V_i(E_c)| = d(d - 1)^{h - 1}$;
  \item for all distinct $u, v \in V_i(E_t) \cup V_i(E_c)$, we have
    $B_{r}(u) \cap B_{r}(v) = \emptyset$.
  \end{itemize}
  
  Additionally, we can find such a set in polynomial time.
\end{proposition}
\begin{proof}
  We will describe the efficient algorithm that does this.

  We are going to incrementally grow our set $E_t$, one edge at the
  time, until $|V_i(E_t) \cup V_i(E_c)| = d(d - 1)^{h - 1}$, so
  suppose $E_t$ is initially an empty set. We start by, for all
  $e = (u, v) \in E_c$, marking all vertices in $B_{1 + r}(\{v,
  u\})$. Note that we marked at most
  $\tau \cdot (d(d - 1)^{r}) \leq 2 \tau (d - 1)^{r + 1}$ vertices.

  Notice that, since we marked all vertices at distance $1 + r$ from
  any vertex in $V_i(E_c)$, we can safely pick any unmarked vertex and
  an arbitrary neighbor and add that edge to $E_t$.

  We can now describe a procedure to add a single edge to $E_t$:
  
  \begin{itemize}
  \item Pick an unmarked vertex $u$ and an arbitrary neighbor $v$ of
    $u$;
  \item Add $(u, v)$ to $E_t$;
  \item Mark all vertices in $B_{1 + r}(\{u, v\})$.
  \end{itemize}

  By the same reasoning as before, as long as we have an unmarked
  vertex, this procedure works. If we repeat the above $t$ times, we
  are left with at least

  \[
    n - 2 \tau (d - 1)^{r + 1} - 2 t (d - 1)^{r + 1}
  \]

  \noindent unmarked vertices. We claim the procedure can be
  successfully repeated at least $2 \tau (d - 1)^{r/2 + 2}$ times. In
  such a case, the number of unmarked vertices left is at least:

  \[
    n - 2 \tau (d - 1)^{r + 1} - 4 \tau (d - 1)^{r/2 + 2} (d -
    1)^{r + 1} \geq  n - 6 \tau (d - 1)^{3r / 2 + 3},
  \]

  \noindent which is always greater than $0$ when
  $r \leq \frac{2}{3}\log_{d - 1}(n / \tau) - 5$. Hence, we always
  have at least one unmarked vertex to pick throughout the procedure.

  Note that the number of repetitions we require exactly matches the
  size of $|E_t|$ so we need this to be exactly
  $d(d - 1)^{h - 1} - \tau \leq 2 \tau (d - 1)^{r/2+2}$, which means
  our algorithm always succeeds.
\end{proof}

We will state some simple properties of this construction that will be
relevant later on.

\begin{fact} \label{fact:vsize}
  $|V_i(E_t)| \geq \tau \cdot (d - 1)^{\ceil{r / 2}}$
\end{fact}
\begin{proof}
  We simply have:
  $|V_i(E_t)| = |E_t| = d(d - 1)^{h - 1} - \tau \geq \tau \cdot (d -
  1)^{\ceil{r / 2}}$.
\end{proof}

\begin{fact} \label{fact:tree} For all $e \in E_t$, there is at most
  one cycle in $B_{r}(e)$ in $G$ and if there is a cycle it has
  length greater than $r$.
\end{fact}
\begin{proof}
  That there is at most one cycle in $B_{r}(e)$ is obvious since $G$
  is bicycle-free at radius $r$. So, let's suppose there is a cycle
  $C$ in $B_{r}(e)$ with length less than or equal to $r$. Then, there
  is at least one edge $e' \in C$ that is also in $E_c$, but in that
  case $B_{r}(e) \cap B_{r}(e') \neq \emptyset$, which contradicts the
  definition of $E_t$.
\end{proof}

We can now extend our definition of $H_t$. Let $H$ be the graph
obtained from $G$ by removing all edges in $E_c$ and in $E_t$.

Recall our plan to add $T_1$ and $T_2$, two $d$-regular trees of
height $h$ (recall $h = \ceil{\log_{d-1} \tau} + \ceil{r / 2} + 1$),
to $H$ while pairing up elements of $L_i$ with endpoints of removed
edges. We will now describe a pairing process that achieves high girth
(and later we will see how it also achieves low $\lambda$).

First, consider a canonical ordering of $L_1$ and $L_2$ based on visit
times from a breath-first search, as illustrated in \Cref{fig:treeord}
for $d = 3$. Given this ordering, the following is easy to see:

\begin{fact} \label{fact:tdist}
  The tree distance between two leaves with indices $i$ and $j$ is at
  least $2(1 + \log_{d-1} (|i - j| + 1) / d)$.
\end{fact}
\begin{proof}
  Let's show that the lowest common ancestor of the two leaves is at
  least $1 + \log_{d-1} |i - j + 1| / d$, this proves the claim since
  we need to travel this distance twice, from the $i$th indexed leaf
  to the ancestor and then back to the $j$th indexed leaf. Let $V_0$
  be the set of $|i - j| + 1$ leaves with indices between $i$ and
  $j$. Let's construct the smallest subtree that includes $V_0$ from
  bottom up and compute its height, which is an upper bound to the
  desired lowest common ancestor. First, group elements of $V_0$ in
  groups of at most $d - 1$ consecutive indices and add one
  representative of each group to a set $V_1$. Each group corresponds
  to a node that parents all of its elements. There are at most
  $|V_0| / (d - 1)$ such groups, so $|V_1| \leq |V_0| / (d -
  1)$. Repeat the same procedure until $|V_a| \leq 1$, in which case
  $a$ is an upper bound to the height of the goal subtree, and by
  induction we have that $|V_{i + 1}| \leq |V_i| / (d - 1)$, so
  $a \geq \log_{d - 1} |V_0|$.

  This is not quite right because if the last grouping corresponds to
  the root of the tree, we need to group elements in $d$ groups,
  because this is the degree of the root, so by accounting for this we
  have $a \geq 1 + \log_{d - 1} (|V_0| / d)$.
\end{proof}

Now, consider the following pairing of elements in $L_1$ and
$V_1(E_t) \cup V_1(E_c)$: pick an arbitrary element of $V_1(E_c)$ and
pair it up with the first leaf of $L_1$. Now pick
$(d - 1)^{\ceil{r / 2}}$ distinct elements of $V_1(E_t)$ and pair them
up with the next leaves of $L_1$. Repeat this procedure, of pairing
one element of $V_1(E_c)$ with $(d - 1)^{\ceil{r / 2}}$ elements of
$V_1(E_t)$ with a contiguous block of leaves until we exhaust all
elements of $V_1(E_c)$. Note that by \Cref{fact:vsize}, there always
are enough elements in $E_t$ to perform this pairing. Pair up any
remaining leaves with the remaining elements of $V_1(E_t)$
arbitrarily. Now repeat the same procedure but for $L_2$ and
$V_2(E_t) \cup V_2(E_c)$ with the same groupings (so the endpoints of
an edge in either $E_t$ or $E_c$ are mapped to the same leaves of
$L_1$ and $L_2$). This pairing procedure is pictured in
\Cref{fig:pairing} below.

\begin{figure}[h]
  \begin{minipage}[b]{0.35\textwidth}
    \centering
    \includegraphics[height=3.8cm]{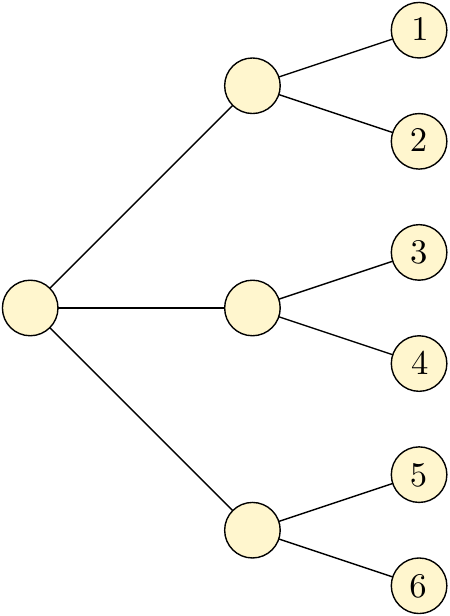}
    \caption{Leaf ordering for $d = 3$}
    \label{fig:treeord}
  \end{minipage}
  \begin{minipage}[b]{0.65\textwidth}
    \centering
    \includegraphics[height=3.8cm]{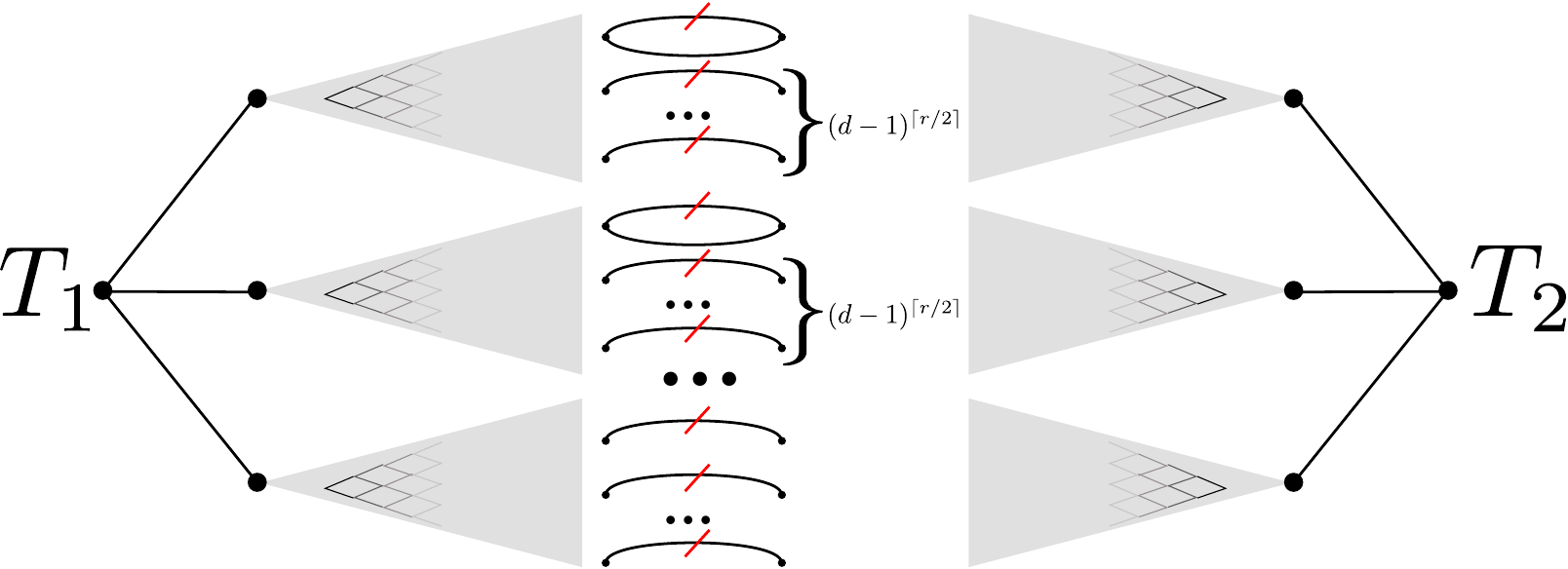}
    \caption{Example pairing}
    \label{fig:pairing}
  \end{minipage}
\end{figure}

Let $\fix(G)$ be defined as the graph resulting from applying the
method described in the previous paragraph to fix the degrees of
$H$. It is now obvious that $\fix(G)$ is a $d$-regular graph and we
only add $|T_1| + |T_2| = O(\tau \cdot (d - 1)^{r / 2 + 1})$ new vertices,
so it has $n + O(\tau \cdot (d - 1)^{r / 2 + 1})$ total vertices. We will
now analyze the resulting girth and $\lambda$ value and prove
\Cref{thm:cycrem} in the process.

\subsection{Analyzing the girth of $\fix(G)$}
\label{sec:girth}

Here we prove that the girth of $\fix(G)$ is at least $r$. Let's start
by supposing, for the sake of contradiction, that there is a cycle $C$
of length less than $r$. We know that the girth of $H$ is more than
$r$ by definition, so $C$ has to use an edge from $T_1$ or
$T_2$. Without loss of generality, let's assume that $C$ contains at
least one edge from $T_1$. Since $T_1$ is a tree, $C$ has to
eventually exit $T_1$ and use some edges from $H$, so in particular it
uses some vertex $v \in L_1$. We will show that in this case, $C$ has
length at least $r$, which is a contradiction. Thus, we have to handle
two cases: $v \in V_1(E_c)$ and $v \in V_1(E_t)$.

Let us start with the $v \in V_1(E_c)$ case. Let's follow $C$ starting
in $v$ and show that to loop back to $v$, $C$ would require to traverse
at least $r$ edges. So, we start in $v$ and go into $T_1$ by following
the only edge in $T_1$ that connects to $v$. Then, the cycle $C$ has
to use some edges from $T_1$ and finally exit through some other
vertex in $L_1$ before eventually looping back to $v$. Suppose that
$u \in L_1$ is such a vertex. Due to our grouping of elements in
$E_t$ with $(d - 1)^{\ceil{r / 2}}$ elements in $E_c$, if $u$ is in
$V_1(E_c)$, we know that the tree indices of $v$ and $u$ differ by at
least $(d - 1)^{\ceil{r / 2}}$. Hence, plugging this into the bound
from \Cref{fact:tdist}, the tree distance between $v$ and $u$ is at
least $r - 1$, which would imply $C$ has length at least $r$. So $u$
has to be in $V_1(E_t)$.

Continuing our traversal of $C$, we now exit $T_1$ through $u$ and
need to loop back to $v$. From our construction in \Cref{prop:treeadd}
we know that the distance in $H$ between $v$ and $u$ is at least $r$,
so any short path in $\fix(G)$ between these vertices has to go
through $T_1$ or $T_2$. Again, our \Cref{prop:treeadd} construction
gives that the distance in $H$ between $v$ and any other vertex in
$L_1$ is at least $r$, so such a short path will have to use some
edges in $T_2$.

Finally, we claim that the distance from $u$ to any vertex $w$ in
$L_2$ is at least $r$. If $w \neq \phi(u)$, we know from our
\Cref{prop:treeadd} construction that the distance between $u$ and $w$
is at least $r$. Otherwise, if there is a path $P$ of length less than
$r$ from $u$ to $w$, then the cycle $P + uw$ has length at most $r$
and is in $B_r(\{u, w\})$, which contradicts \Cref{fact:tree}. In
conclusion, it is not possible to loop back to $v$ using less than $r$
steps, which concludes the proof of the $v \in V_1(E_c)$ case.

The proof for the $v \in V_1(E_t)$ case is already embedded in the
previous proof, so we will just sketch it. Using the same argument we
start by following $C$ into $T_1$ and eventually exiting through some
vertex $u \in V_1(E_t)$. As we saw before, the $H$ distance between
$u$ and $v$ is at least $r$ and the $H$ distance between $u$ and any
other vertex in $L_1$ or any vertex in $L_2$ is at least $r$, so we
cannot loop back to $v$ from $u$, which concludes the proof of this
case.

\subsection{Bounding $\lambda(\fix(G))$}

We finally analyze the spectrum of $\fix(G)$ by proving that
$\lambda(\fix(G)) \leq \Lambda + O(1 / r)$. This argument is very
similar to the proof in Section 4 of \cite{AGS19}, but adapted to our
construction.

First, observe that the adjacency matrix of $\fix(G)$, which we will
denote by simply $A$, can be written in the following way:
$A = A_G - A_{E_c} - A_{E_t} + A_{T_1} + A_{T_2}$, where $A_G$ is the
adjacency matrix of $G$ defined on the vertex set of $\fix(G)$ (which
is to say $G$ with a few isolated vertices from the added trees),
$A_{E_c}$ is the adjacency matrix of the cycle edges removed, and so
on. Also, let $V_G$ be the set of vertices from $G$, $V_1$ the set of
vertices from $T_1$ and $V_2$ the set of vertices from $T_2$, so
$V = V_G \cup V_1 \cup V_2$. In this section we will prove
$\lambda(A) \leq \Lambda + O(1 / r)$.

Let $g$ be any unit eigenvector of $A$ orthogonal to the all ones
vector, so $\sum_{v \in V} g_v^2 = 1$ and $\sum_{v \in V} g_v = 0$. We
have that $|\sum_{v \in V_1 \cup V_2} g_v| \leq \sqrt{2 |T_i|}$ by
Cauchy-Schwarz (since this vector is supported on only $2 |T_i|$
entries), which in turn implies that
$|\sum_{v \in V_G} g_v| \leq \sqrt{2 |T_i|}$.

It suffices to show that $|g^TAg| \leq \Lambda + O(1 / r)$. To
do so, we shall analyze the contributions of $A_G$, $A_{E_c}$,
$A_{E_t}$, $A_{T_1}$ and $A_{T_2}$ to $|g^TAg|$.

To bound the contribution of $A_{T_1}$ and $A_{T_2}$, we use a lemma
proved by Alon-Ganguly-Srivastava:

\begin{lemma}(\cite[Lemma.~4.1]{AGS19}). \label{lem:trees} Let $W_i$ be the set of
  non-leaf vertices of $T_i$. Then for any vector $f$ we have:

  \[
    |f^TA_{T_i}f| \leq 2\sqrt{d - 1}\sum_{w \in W_i} f_w^2 + \sqrt{d - 1}\sum_{v \in L_i} f_v^2.
  \]
\end{lemma}

Recall that the edges in $E_t \cup E_c$ define a perfect matching
between $L_1$ and $L_2$, so we have the following:

\[
  |g^T(A_{E_c} + A_{E_t})g| = \left |\sum_{uv \in E_t \cup E_c}
    2g_ug_v \right | \leq \sum_{v \in L_1 \cup L_2} g_v^2.
\]

Finally, let $g_G$ be the projection of $g$ to the subspace spanned by
$V_G$. Observe that $|g^TA_Gg| = |g_G^TA_Gg_G|$. Now, let $\bone_G$ be
the all ones vector supported on the set $V_G$ and $g_\perp$ be a
vector orthogonal to $\bone_G$ such that $g_G = a\bone_G + g_\perp$,
for some constant $a$. We have that
$\bone_G^T g_G = a \bone_G^T \bone_G$, which implies

\[
  |a| = \left | \frac{\sum_{v \in V_G} (g_G)_v}{n} \right | \leq
  \frac{\sqrt{2 |T_i|}}{n}.
\]

Now observe:

\[
  |g_G^TA_Gg_G| \leq |g_\perp^TA_Gg_\perp| + |(a\bone_G)^TA_G(a\bone_G)| \leq \Lambda \sum_{v \in V_G}g_v^2 + \frac{2|T_i|d}{n},
\]

\noindent where the $\Lambda$ bound on $|g_\perp^TA_Gg_\perp|$ comes
from the definition of $G$. Also, we claim that the term
$\frac{2|T_i|d}{n}$ is $O(1 / r)$. We have
$|T_i| = O(\tau \cdot (d-1)^{r / 2 + 1})$ and we know from the problem
constraints that $r \leq (2/3)\log_{d - 1}(n / \tau) - 5$ which
implies
$\tau \cdot (d-1)^{r / 2 + 1} / n \leq O((d - 1)^{-r}) = O(1 / r)$.

We can now plug everything together and apply \Cref{lem:trees} to
obtain:

\[
  |g^TAg| \leq \Lambda + (\sqrt{d - 1} + 1)\sum_{v \in L_1 \cup L_2} g_v^2 + O(1 / r).
\]

We will conclude our proof by showing that
$\sum_{v \in L_1 \cup L_2} g_v^2$ is $O(1 / r)$. It should be clear
from the symmetry of our construction that we only need to prove
$\sum_{v \in L_1} g_v^2 = O(1 / r)$, since the same is analogous for
$L_2$.

Similarly to what was done in \cite{AGS19}, we are going to use a
theorem proved by Kahale, originally used to construct Ramanujan
graphs with better expansion of sublinear sized subsets.

\begin{lemma}(\cite[Lemma~5.1]{Kah95}). \label{lem:kah}
  Let $v$ be some vertex of
  $V$. Let $l$ be a positive integer and $s$ some vector supported on
  $V$. Let $X_i$ be the set of all vertices at distance exactly $i$
  from $v$ in $\fix(G)$. Assume that the following conditions hold:
  
  \begin{enumerate}
  \item For $l - 1 \leq i,j \leq l$, all vertices in $X_i$ have the
    same number of neighbors in $X_j$.
  \item The vector $s$ is constant on $X_{l - 1}$ and on $X_l$.
  \item The vector $s$ has positive components and
    $(As)_u \leq |\mu| s_u$ for all $u \in B_{l-1}$, where $|\mu|$ is a
    non-zero real number.
  \end{enumerate}

  Then, for any vector $y$ supported on $V$ satisfying
  $|(Ay)_u| = |\mu| |y_u|$ for $u \in B_{l - 1}$ we have:

  \[
    \frac{\sum_{u \in X_l} y_u^2}{\sum_{u \in X_l} s_u^2} \geq \frac{\sum_{u \in X_{l - 1}} y_u^2}{\sum_{u \in X_{l - 1}} s_u^2}.
  \]
\end{lemma}

Our plan is to pick the parameters $l, s$ and $v$ from \Cref{lem:kah}
and use it to show that $\sum_{v \in L_1} g_v^2 = O(1 / r)$. Let $\mu$
be the eigenvalue associated with $g$ and suppose that
$|\mu| > 2\sqrt{d - 1}$, otherwise $|\mu| \leq \Lambda$, which would
imply the result. Set $v$ to be the root of $T_1$ and set
$s_u = (d-1)^{-i/2}$ if $u \in X_i$. We claim that this choice of $v$
and $s$ satisfies the conditions of \Cref{lem:kah} for all
$l \leq h + \floor{r / 2}$, where
$h = \ceil{\log_{d-1} \tau} + \ceil{r / 2} + 1$ is the height of $T_1$
and $T_2$. Let's prove this for each of the conditions in the above
order:

\begin{enumerate}
\item For $l = 0$, we have that $X_l = \{v\}$, and $v$ has $d$
  neighbors in $X_1$ and no other neighbors.

  For $1 \leq l \leq h$, we have that $X_l$ is the $l$th level of
  $T_1$, so it is clear that any vertex in $X_l$ has one neighbor in
  $X_{l-1}$, $d-1$ in $X_{l+1}$, and no neighbors in $X_l$.

  For $h < l \leq h + \floor{r / 2}$, we claim that a vertex in $X_l$
  also has $d - 1$ neighbors in $X_{l + 1}$ and one neighbor in
  $X_{l - 1}$. First, note that $X_{h} = L_1$ and so let $u \in L_1$
  be some vertex. We first prove the following proposition, whose
  proof uses some of the ideas of \Cref{sec:girth}:

  \begin{proposition} \label{prop:texit} Let $u$ be a vertex in
    $L_1$. Let $\calP(u)$ be the set of non-empty paths that start in
    $u$ and whose first step does not go into $T_1$. Then, the
    shortest path in $\calP(u)$ that ends in any vertex in $L_1$ has
    length at least $r$.
  \end{proposition}
  \begin{proof}
    As in the previous girth proof, we have two cases,
    $u \in V_1(E_c)$ and $u \in V_1(E_t)$. The latter case is obvious
    from the proof in \Cref{sec:girth}, since if $u \in V_1(E_t)$ then
    the $H$ distance to any node in $L_1$ is at least $r$ (from
    \Cref{prop:treeadd}) and the $H$ distance to any node in $L_2$ is
    also at least $r$ (from \Cref{fact:tree}). So, suppose
    $u \in V_1(E_c)$.

    Let's follow the same proof strategy as before, so let
    $P \in \calP(u)$ be the shortest path and let's follow $P$
    starting in $u$. Again, from \Cref{prop:treeadd} the $H$ distance
    of $u$ to any node in $L_1$ is at least $r$. However, $u$ might
    reach $\phi(u)$ in a short number of steps (namely, if the cycle
    corresponding to $(u, \phi(u))$ is short). So, let's follow $P$ to
    $\phi(u)$ and into $T_2$. We are now in the exact same situation
    as in the setup of the proof in \Cref{sec:girth} (but starting in
    $T_2$), so the result follows.
  \end{proof}

  Let's say a vertex $w$ is at $\calP$-distance $\delta$ from $u$ if
  the shortest path $P \in \calP(u)$ that ends in $w$ has length
  $\delta$. Additionally, let $S_\delta(u)$ be the set of vertices
  that are at a $\calP$-distance of at most $\delta$ from $u$. From
  \Cref{prop:texit}, we know that for all distinct $u, w \in L_1$, the
  sets $S_{\floor{r / 2}}(u)$ and $S_{\floor{r / 2}}(w)$ are
  disjoint. Thus, we have that for $u \in L_1$ the vertices in
  $S_{\floor{r / 2}}(u)$ form a tree rooted at $u$. So we can conclude
  by using the same argument as in the $1 \leq l \leq h$ case.
\item By definition of $s$, condition 2 is true.
\item For $1 \leq l \leq h + \floor{r / 2}$ and $u \in X_l$, we have:
  \[
    (As)_u = \sum_{w \sim u} s_w = (d-1)^{-(l+1)/2}(d - 1) + (d-1)^{-(l - 1)/2} = 2\sqrt{d-1} \cdot (d - 1)^{-i/2} \leq |\mu| s_u.
  \]

  For the special case of $l = 0$ we trivially have
  $(As)_v = d \cdot (d-1)^{-1/2} \leq |\mu| s_v$.
\end{enumerate}

We can now apply \Cref{lem:kah} and conclude that for all
$1 \leq l \leq h + \floor{r / 2}$, we have
$\sum_{u \in X_l} g_u^2 \geq \sum_{u \in X_{l - 1}} g_u^2$, since for
such $l$ we have $\sum_{u \in X_l} s_u^2 = d / (d - 1)$. So the
sequence $(\sum_{u \in X_l} g_u^2)_l$ is an increasing
sequence. Recall that $X_{h} = L_1$, so
$\sum_{u \in X_{h}} g_u^2 = \sum_{u \in L_1} g_u^2$. Additionally, we
know that the total sum of $(\sum_{u \in X_l} g_u^2)_l$ is at most one
(since $g$ is a unit vector), so we have that
$\sum_{l = h}^{h + \floor{r / 2}} \sum_{u \in X_l} g_u^2 \leq \floor{r / 2} \cdot
\sum_{u \in X_{h}} g_u^2 \leq 1$ and finally
$\sum_{u \in L_1} g_u^2 = \sum_{u \in X_{h}} g_u^2 \leq 1 / \floor{r / 2} =
O(1 / r)$.

This concludes the proof of \Cref{thm:cycrem}.

\section{A near-Ramanujan graph distribution of girth
  $\Omega(\log_{d - 1} N)$} \label{sec:apprand}

Recall \Cref{thm:bor}, which says that uniformly random $d$-regular
graphs are near-Ramanujan. We will combine this result with our
machinery of \Cref{sec:rem} to show \Cref{thm:apprand}, namely that
there exists a distribution over graphs that is
$(2\sqrt{d - 1} + \epsilon, c\log_{d-1} n)$-good for any
$\epsilon > 0$ and $c < 1/4$, which we will show is the distribution
resulting from applying algorithm $\fix$ to a sample of $\calG_d(n)$.

First, we note that $\calG_d$ has nice bicycle-freeness. We quote the
relevant result from \cite{Bor19}, which we reproduce below:

\begin{lemma}(\cite[Lemma~9]{Bor19}).  \label{lem:bfree} Let
  $d \geq 3$ and $r$ be positive integers. Then $G \sim \calG_d(n)$ is
  bicycle-free at radius $r$ with probability
  $1 - O((d - 1)^{4r} / n)$.
\end{lemma}

An obvious corollary of this is that for any constant $c < 1/4$, we
have that $G \sim \calG_d(n)$ is bicycle free at radius
$c \log_{d - 1} n$ with high probability.

To bound the number of short cycles in $\calG_d(n)$ we use a classic
result that very accurately estimates the number of short cycles in
random regular graphs.

\begin{lemma}(\cite[Section~2]{kay04}).  \label{lem:shortnm} Let
  $G \sim \calG_d(n)$ and $X_i$ be the random variable that denotes
  the number of cycles of length $i$ in $G$. Let
  $R_i = \max\{(d - 1)^i / i, \log n\}$. Then
  
  \[
    \prob{X_i \leq R_i, \text{ for all }3 \leq i \leq 1 / 4 \log_{d - 1} n} = 1 - o_n(1).
  \]
\end{lemma}

Given the above, we obtain the following bound, for all $c < 1/4$:

\[
  \sum_{i=1}^{c\log_{d-1}n}{\max\{(d - 1)^i / i, \log n\}} = O(n^c).
\]

So we obtain the following proposition:

\begin{proposition}
  For any $c < 1/4$ and any $\epsilon > 0$, $G \sim \calG_d(n)$ is a
  $(c \log_{d - 1} n, 2\sqrt{d - 1} + \epsilon, O(n^c))$-graph with
  probability $1 - o_n(1)$.
\end{proposition}

Finally, we want to apply \Cref{thm:cycrem}, so first we need to
verify its preconditions. For all $c < 1/4$ we have that
$(2/3) \log_{d-1} (n / n^c) = (2/3)(1 - c) \log_{d-1} n \leq c
\log_{d-1} n$. Also note that
$n^c (d - 1)^{c/2 \log_{d-1} n + 1} = n^{3c/2} = O(n^{3/8})$, so when
applying \Cref{thm:cycrem} the resulting graph has
$n + O(n^{3/8}) = n(1 + o_n(1))$ vertices. Thus, we obtain
\Cref{thm:apprand}.

\begin{remark} \label{rem:conf}
  Recall that $\calG_d(n)$ is the same as the conditional distribution
  of the $d$-regular $n$-vertex configuration model when conditioned
  on it being a simple graph. Indeed, a graph drawn from the
  $d$-regular $n$-vertex configuration model is simple with
  probability $\Omega_d(1)$. A result very similar to
  \Cref{lem:shortnm} also holds for the configuration model and thus
  the results of this section also hold for the configuration model.
\end{remark}

\subsection{Counting near-Ramanujan graphs with high girth}

We will briefly prove \Cref{cor:apprand} using the result we just
proved. For simplicity, we are going to work with the configuration
model, using the observation of \Cref{rem:conf}.

Our proof will use a classic result on the number of not necessarily
simple $d$-regular $n$-vertex graphs, which is the same as the number
of graphs in the $n$-vertex $d$-regular configuration model. It is
easy to show \cite{BC78} that for $nd$ even, the number of such graphs
is

\[
  \sim \parens*{\parens*{\frac{d^dn^d}{e^d(d!)^2}}^{n / 2}}.
\]

Hence, the core claim we need to prove, is the following:

\begin{proposition} \label{prop:disfix} Let $G_1$ and $G_2$ be
  distinct graphs that follow the preconditions of
  \Cref{thm:cycrem}. Then $\fix(G_1)$ and $\fix(G_2)$ are also
  distinct.
\end{proposition}

This proposition implies that given any two good $d$-regular
$n$-vertex graphs, applying $\fix$ produces two distinct graphs. From
our proof of \Cref{thm:apprand} we also know that the result of
applying $\fix$ adds at most $O(n^{3/8})$ vertices. Finally, since a
$(1 - o_n(1))$ fraction of the graphs are good an thus when we apply
$\fix$ they result in ($2\sqrt{d - 1} + \epsilon$,
$c \log_{d - 1} n$)-good graphs, the result follows.

\begin{proof}[Proof sketch of \Cref{prop:disfix}]  
  Recall the $H$ graph from the description of $\fix$ and let $H_1$ be
  such graph corresponding to $G_1$ and define $H_2$
  analogously. Let's suppose for the sake of contradiction that
  $\fix(G_1)$ and $\fix(G_2)$ are isomorphic, then we have that $H_1$
  is isomorphic to $H_2$. If this is the case, then let $S$ be the set
  of vertices whose degrees are $d - 1$ and hence they had edges
  removed that were part of cycles (note that we can remove multiple
  edges adjacent to one vertex since this would imply the existence of
  two cycles in a small neighborhood, breaking the bicycle-freeness
  assumption). Now, the edges removed from $G_1$ (similarly $G_2$)
  form a perfect matching on $S$ that adds exactly $|S| / 2$ cycles to
  $H$. We conclude by pointing out that there is exactly one perfect
  matching that adds $|S| / 2$ cycles to $H$, which means the same
  edges were removed from $G_1$ and $G_2$, which would imply they are
  the same.
\end{proof}

\section{Explicit near-Ramanujan graphs of girth $\Omega(\sqrt{\log n})$} \label{sec:appdet}

In this section we prove \Cref{thm:appdet}, building on the
construction in the proof of \Cref{thm:mop}. We note that the original
construction has no guarantees on the girth of the constructed graph
other than a constant girth. We will briefly recap the main tools and
ideas from the paper.

\subsection{Review of constructing explicit near-Ramanujan graphs}

Given a $d$-regular $n$-vertex graph $G = (V, E)$, let
$w \in \{\pm 1\}^E$ be an edge-signing of $G$. The \textit{2-lift of $G$
  given $w$} is defined as the following $d$-regular $2n$-vertex graph
$G_2 = (V_2, E_2)$:

\[
  V_2 = V \times \{\pm 1\} \qquad E_2 = \left \{\{(u, \sigma), (v, \sigma
  \cdot w(u, v))\} : (u, v) \in E, \sigma \in \{\pm 1\} \right \}.
\]

It was observed in \cite{BL06} that the spectrum of $G_2$ is given by
the union of the spectra of $G$ and $\widetilde{G}_w$, where the
latter refers to the eigenvalues of the adjacency matrix of $G$ signed
according to $w$, where each nonzero entry is $w(u, v)$ for
$\{u, v\} \in E$.

This connection between the spectrum of an edge-signing of a graph and
a 2-lift gave rise to the following theorem, which was proved in
\cite{MOP19b}. Below we write
$\rho(G) = \max\{|\lambda_i| : i \in [n]\}$ for the \textit{spectral
  radius} of $G$.

\begin{theorem}(\cite[Theorem~3.1]{MOP19b}). \label{thm:mainmop}
  Let $G = (V, E)$ be an arbitrary $d$-regular $n$-vertex graph
  ($d \geq 3$). Assume $G$ is bicycle-free at radius
  $r \gg (\log\log n)^2$. Then for a uniformly random edge-signing
  $w$, except with probability at most $n^{-100}$ we have:

  \[
    \rho(\widetilde{G}_w) \leq 2\sqrt{d - 1} \cdot \left ( 1 + \frac{(\log\log n)^4}{r^2}\right ).
  \]

  Furthermore, this can be derandomized: given a constant $C$ there is
  a generator $h : \{0, 1\}^s \to \{\pm 1\}^E$ computable in time
  $\poly(N^{C \log d})$, with seed length
  $s = O(\log(2C) + \log\log n + C \cdot \log(d) \cdot \log(n))$, such
  that for $u \in \{0, 1\}^s$ chosen uniformly at random, with
  probability at most $n^{-100}$ we have:

  \[
    \rho(\widetilde{G}_{h(u)}) \leq 2\sqrt{d - 1} \cdot \left ( 1 +
      \frac{(\log\log n)^4}{r^2}\right ) + \frac{\sqrt{d}}{C^2}.
  \]
\end{theorem}

This theorem is a powerful tool that, combined with the above
observation, allows one to double the number of vertices in a
near-Ramanujan graph while keeping it near-Ramanujan, as long as the
bicycle-freeness is good enough. It is easy to show that if $G$ is
bicycle-free at radius $r$, then \textit{any} 2-lift of $G$ is also
bicycle-free at radius $r$. So, the strategy employed by \cite{MOP19b}
is to start with a graph with a smaller number of vertices that is
bicycle-free at a big enough radius and 2-lift it enough times until
the graph has the required number of vertices.

To generate this starting graph, the authors first showed how out to
weakly derandomize \cite{Bor19}. Formally, the following is proved:

\begin{theorem}(\cite[Theorem~4.8]{MOP19b}).
  \label{thm:rebordenave}
  For a large enough universal constant $\alpha$ and any integer
  $n > 0$, given $d$, $\epsilon$ and $c$ such that:

  \[
    3 \leq d \leq \alpha^{-1}\sqrt{\log n}, \; \alpha^3\cdot
    \parens*{\frac{\log\log n}{\log_{d-1} n}}^2 \leq \epsilon \leq 1, \; c < 1/4.
  \]
  
  Let $G$ be chosen from the $d$-regular $n$-vertex uniform
  configuration model. Then, except with probability at most
  $n^{-.99}$, the following hold:

  \begin{itemize}
  \item $G$ is bicycle-free at radius $c \log_{d-1} n$;
  \item $\lambda(G) \leq 2\sqrt{d - 1} \cdot (1 + \epsilon)$;
  \end{itemize}

  Furthermore, this can be derandomized: there is a generator
  $h : \{0, 1\}^s \to \calG_d(n)$, with seed length
  $s = O(\log^2(n) / \sqrt{\epsilon})$ computable in time
  $\poly(n^{\log (n) / \sqrt{\epsilon}})$, such that for
  $u \in \{0, 1\}^s$ chosen uniformly at random, with probability at
  most $n^{-.99}$ we have that the above statements remain true for
  $G = h(u)$.
\end{theorem}

Using these two theorems we can setup the construction of
\cite{MOP19b}. So, first assume we are given $n, d \geq 3$ and
$\epsilon > 0$ and we wish to construct a $d$-regular graph $G$ with
$~n$ vertices with $\lambda(G) \leq 2\sqrt{d - 1} + \epsilon$. The
construction is now the following:

\begin{enumerate}
\item Use \Cref{thm:rebordenave} to construct a $d$-regular graph
  $G_0$ with a small number of vertices $n_0 = n_0(n)$. If we pick
  $n_0$ to be $2^{O(\sqrt{\log n})}$ then the generator seed length is
  $O(\log(n) / \sqrt{\epsilon})$ and is computable in time
  $\poly(n^{1/\sqrt{\epsilon}})$, so we can enumerate over all
  possible seeds and find at least one that produces a graph that is
  bicycle-free at radius
  $\Omega(\log(n_0)) = \Omega(\sqrt{\log n}) \gg (\log\log n)^2$ and
  has $\lambda(G_0) \leq 2\sqrt{d - 1} \cdot (1 + \epsilon)$ in
  $\poly(n)$ time.
\item Next, we can repeatedly apply \Cref{thm:mainmop} to double the
  number of vertices of $G_0$, by choosing $C$ to be
  $\sim d^{1/4}/\sqrt{\epsilon}$. We then enumerate over all seeds
  until we find one that produces a good graph, which only requires
  $\poly(n)$ time. On each application the bicycle-freeness radius is
  maintained (so we can keep applying \Cref{thm:mainmop}) and the
  number of vertices of doubles. After roughly $\log(n/n_0)$
  applications, the resulting graph has $n(1 + o_n(1))$ vertices and
  $\lambda(G) \leq 2\sqrt{d - 1} \cdot (1 + \epsilon)$.
\end{enumerate}

\subsection{Improving the girth of the construction}

We are finally ready to prove \Cref{thm:appdet}. We are going to apply
a similar strategy as the one from \Cref{sec:apprand}. Instead of
derandomizing \Cref{lem:shortnm} we are going to obtain a simpler
bound, which is good enough to obtain the desired. We note however,
that \Cref{lem:shortnm} can be derandomized and for completeness we
show how to in \Cref{sec:drand}.

We start by proving the following lemma:

\begin{lemma} \label{lem:bftogt} Let $G$ be a $d$-regular $n$-vertex
  graph with $\lambda(G) = \Lambda \geq 2\sqrt{d - 1}$ and such that
  $G$ is bicycle-free at radius $\alpha \log_{d - 1} n$, for
  $\alpha \leq 2$. Then we can apply $\fix$ to $G$ and obtain a graph
  such that:
  \begin{itemize}
  \item $\fix(G)$ is $d$-regular and has $n(1 + o_n(1))$ vertices;
  \item $\lambda(\fix(G)) \leq \Lambda + o_n(1)$;
  \item $\fix(G)$ has girth $(\alpha / 3) \log_{d - 1} n$.
  \end{itemize}
\end{lemma}

Before proving this lemma, we prove a core proposition in a slightly
more generic way.

\begin{proposition} \label{prop:breetocyc}
  Let $G$ be a $d$-regular graph that is bicycle-free at radius $2r$,
  then

  \[
    |\Cyc_{r}(G)| \leq n / (d - 1)^{r}.
  \]
\end{proposition}
\begin{proof}
  Pick one vertex per cycle in $\Cyc_{r}(G)$ and place it in a set
  $S$. We claim that for every distinct $u, v \in S$,
  $B_{r}(u) \cap B_{r}(v) = \emptyset$. Suppose this wasn't
  the case and suppose there is some $w$ such that
  $w \in B_{r}(u) \cap B_{r}(v)$, for some pair $u, v$. Note
  that $B_{2r}(w)$ includes the two length $r$ cycles that correspond
  to $u$ and $v$, which contradicts bicycle-freeness in $G$.

  Given the above, we have that the sets $B_{r}(u)$ for $u \in S$
  are pairwise disjoint and also we know that
  $|B_{r}(u)| = d(d - 1)^{r - 1}$. Hence we have:

  \[
    |\Cyc_{r}(G)| \cdot d(d - 1)^{r - 1} \leq n,
  \]

  \noindent which implies the desired result.
\end{proof}

And we can prove the above lemma.

\begin{proof}[Proof of \Cref{lem:bftogt}]
  By plugging $G$ into \Cref{prop:breetocyc} we can conclude that $G$
  is a $(\alpha \log_{d - 1} n, \Lambda, n^{1 - \alpha/2})$-graph. We
  wish to apply \Cref{thm:cycrem} so first recall its
  preconditions. By definition $\Lambda \geq 2\sqrt{d - 1}$. However,
  the precondition on the radius of bicycle-freeness does not hold,
  since
  $(2/3) \log_{d-1} (n / n^{1-\alpha/2}) = (\alpha / 3) \log_{d-1} n$
  which is less than $\alpha \log_{d-1} n$. If we instead use the
  fact that $G$ is also trivially a
  $((\alpha / 3) \log_{d - 1} n, \Lambda, n^{1 - \alpha/2})$-graph,
  then the precondition is satisfied.

  Thus, we can apply \Cref{thm:cycrem} and we obtain that $\fix(G)$
  satisfies all the required conditions, which concludes the proof.
\end{proof}

Given this lemma, we will modify the first step of the construction of
\cite{MOP19b} to produce a graph $G_0$ with girth $c\sqrt{\log n}$,
for any $c$. Note that, similarly to bicycle-freeness, the girth of a
graph can only increase when applying any 2-lift, so this strategy
guarantees that after step 2 of the construction, the final graph has
the desired girth, which would imply \Cref{thm:appdet}.

First, when enumerating over all seeds to generate $G_0$ in step 1,
let's look for one that guarantees that $G_0$ is bicycle-free at
radius $(1/5) \log_{d - 1} n_0$ (recall that by \Cref{thm:rebordenave}
a $1 - o_n(1)$ fraction of the seeds satisfy this). Next, let's apply
\Cref{lem:bftogt} and obtain that $\fix(G_0)$ has girth
$(1/15) \log_{d - 1} n_0$ and the desired value of $\lambda(G_0)$.

Now, let $\kappa = 15c / \log_{d - 1} 2$. As long as
$\kappa \sqrt{\log n} \leq \log n$ we can set $n_0$ to
$2^{\kappa\sqrt{\log n}}$, in which case $G_0$ has girth
$c \sqrt{\log n}$, but the first step still runs in $\poly(n)$ time
and so we are done.

\section{Open problems}

\begin{itemize}
\item Can we improve \Cref{thm:appdet} to obtain high girth?

  Something like this could be proved by showing that when 2-lifting a
  graph with large enough girth, with sufficiently high
  probability the girth of the resulting graph increases. This would
  boost the girth of the graph generated by the first step of the
  construction of \cite{MOP19b} during the repeated 2-lift
  step. However, it is unclear if this can be done. Alternatively, one
  could show that bicycle-freeness increases with good probability as
  we 2-lift, but this is also unclear.

  A different strategy would be to find a different way to derandomize
  \Cref{thm:bor} such that the we can generate a starter graph of
  larger size. However, it is unclear if this strategy could work
  since the tool used to derandomize this, namely $(\delta,k)$-wise
  uniform permutations (defined in \Cref{sec:drand}), cannot be
  improved to derandomize this to the required extent.
\item Can we obtain \Cref{thm:apprand} for higher values of $c$; for
  example, can we build a distribution that is
  ($2\sqrt{d - 1} + \epsilon$, $.99 \log_{d - 1} n$)-good?

  One promising strategy would be to show that the graphs produced by
  the distribution described in \cite{LS19}, which were shown to have
  girth at least $.99 \log_{d - 1} n$ with high probability, are also
  near-Ramanujan with high probability. Numerical calculations seem to
  indicate that the answer is positive, as pointed out in one of the
  open problems given in \cite{LS19}.
\end{itemize}

\section*{Acknowledgments} I am very grateful to Ryan O'Donnell for
numerous comments and suggestions, as well as very thorough feedback
on an earlier draft of this paper.

\bibliographystyle{alpha}
\bibliography{ramanujan}

\appendix

\section{Derandomizing the number of short cycles} \label{sec:drand}

To make the statement of this section more precise, we will first
define a known derandomization tool.

\begin{definition}[$(\delta,k)$-wise uniform permutations]
  Let $\delta \in [0,1]$ and $k \in \N^+$.  Let $[n]_k$ denote the set
  of all sequences of~$k$ distinct indices from~$[n]$.  A random
  permutation $\pi \in S_{n}$ is said to be
  \emph{$(\delta,k)$-wise uniform} if, for every sequence
  $(i_1, \dots, i_k) \in [n]_k$, the distribution of
  $(\pi(i_1), \dots, \pi(i_k))$ is $\delta$-close in total variation
  distance from the uniform distribution on~$[n]_k$. When
  $\delta = 0$, we simply say that the permutation is \emph{(truly)
    $k$-wise uniform}.
\end{definition}

Kassabov~\cite{Kas07} and Kaplan--Naor--Reingold~\cite{KNR09}
independently obtained a deterministic construction of
$(\delta,k)$-wise uniform permutations with seed length
$O(k \log n + \log(1/\delta))$.

\begin{theorem} \label{thm:knr} (\cite{KNR09,Kas07}).  There is a
  deterministic algorithm that, given $\delta$, $k$, and $n$, runs in
  time $\poly(n^k/\delta)$ and outputs a multiset
  $\Pi \subseteq S_{n}$ (closed under inverses) of cardinality
  $S = \poly(n^k/\delta)$ (a power of~$2$) such that, for
  $\pi \sim \Pi$ chosen uniformly at random, $\pi$ is a
  $(\delta,k)$-wise uniform permutation.
\end{theorem}

This theorem is required to obtain the generator mentioned in
\Cref{thm:rebordenave} and is the reason why $(\delta,k)$-wise uniform
permutations are useful tools to apply here. We will also need a
convenient theorem of Alon and Lovett~\cite{AL13}:

\begin{theorem} \label{thm:al13} (\cite{AL13}). Let
  $\bpi \in S_n$ be a $(\delta,k)$-wise uniform permutation.
  Then one can define a (truly) $k$-wise uniform
  permutation~$\bpi' \in S_n$ such that the total variation
  distance between $\bpi$ and~$\bpi'$ is~$O(\delta n^{4k})$.
\end{theorem}

We can now define a ``derandomized'' version of the configuration
model, using this tool.

\begin{definition}\label{def:kwcm}
  Recall how the configuration model is defined by a perfect matching
  of a set $[nd]$ of ``half-edges''.
  
  Let's denote this matching by $M$ and define a way to generate it
  using random permutations. First a uniformly random permutation
  $\pi \in S_{nd}$ is chosen; then we set
  $M_{\pi(j),\pi(j+1)} = M_{\pi(j+1),\pi(j)} = 1$ for each odd
  $j \in [nd]$.

  We can write the adjacency matrix $A$ of $G$ as the sum, over all
  $i,i' \in [d]$, of $M_{(v,i),(v',i')}$. Hence
  \[
    \bA_{v,v'} = \sum_{i,i' = 1}^d \sum_{\substack{\text{odd} \\ j \in
        [nd]}} (1[\pi(j) = (v,i)] \cdot 1[\pi(j+1) = (v',i')] +
    1[\pi(j) = (v',i')] \cdot 1[\pi(j+1) = (v,i)]).
  \]

  The \textit{$d$-regular $n$-vertex $(\delta, k)$-wise uniform
    configuration model} is defined by using $(\delta, k)$-wise
  uniform permutations instead. Similarly, we define the
  \textit{$d$-regular $n$-vertex $k$-wise uniform configuration
    model}.
\end{definition}

We can now describe the proposition we wish to prove.

\begin{proposition}
  Fix $d \geq 3$, $n$ and $k \geq c\log_{d - 1} n$, where $c <
  1/4$. Let $G$ be drawn from the $d$-regular $n$-vertex $4k$-wise
  configuration model and $X_i$ be the random variable that denotes
  the number of cycles of length $i$ in $G$. Let
  $R_i = \max\{(d - 1)^i / i, \log n\}$. Then
  
  \[
    \prob{X_i \leq R_i, \text{ for all }1 \leq i \leq 1 / 4 \log_{d - 1} n} = 1 - o_n(1).
  \]

  By \Cref{thm:al13}, these statements remain true in the $(\delta, 4k)$-wise
  uniform versions of the model, $\delta \leq 1/n^{16k + 1}$.
\end{proposition}
\begin{proof}
  The proof follows almost directly from the proof of
  \Cref{lem:shortnm}. First, note that $X_i$ can be written as a
  polynomial of degree at most $i$ in the entries of $G$'s adjacency
  matrix, by summing over the products of the edge indicators of all
  possible cycles of length $i$ in $G$. Thus, from our formula in
  \Cref{def:kwcm}, it can be written as a polynomial of degree at most
  $2k$ in the permutation indicators $1[\pi(j) = (v,i)]$. So we can
  compute $\expt{X_i}$ assuming that $X_i$ is drawn from the fully
  uniform configuration model. Similarly, $X_i^2$ can be written as a
  polynomial of degree at most $4k$ in the permutation indicators, so
  we can compute $\vart{X_i}$ assuming that $X_i$ is drawn from the
  fully uniform configuration model.

  From \cite{kay04} we have the following estimates, that only apply
  when $(d - 1)^{2i-1} = o(n)$:

  \[
    \expt{X_i} = \frac{(d - 1)^i}{2i}(1 + O(i(i + d)/n)) \qquad \vart{X_i}
    = \expt{X_i} + O(i(i + d)/n)\expt{X_i}^2.
  \]

  By applying Chebyshev's inequality to each $X_i$, just like in
  \cite{kay04}, we get the desired result.
\end{proof}

We can finally rewrite \Cref{thm:rebordenave} in the language of the
$d$-regular $n$-vertex $(\delta, k)$-wise uniform configuration model
and tack on the result we just proved.

\begin{theorem}
  For a large enough universal constant $\alpha$ and any integer
  $n > 0$, fix $3 \leq d \leq \alpha^{-1}\sqrt{\log n}$ and $c < 1/4$,
  and let $\eps \leq 1$ and~$k$ satisfy
  
  \[
    \eps \geq \alpha^3 \cdot \parens*{\frac{\log \log n}{\log_{d-1}
        n}}^2, \qquad k \geq \alpha \log(n)/\sqrt{\eps}.
  \]
    
  Let $G$ be chosen from the $d$-regular $n$-vertex $k$-wise uniform
  configuration model. Then except with probability at most
  $1/n^{.99}$, the following hold:
  
  \begin{itemize}
  \item $G$ is bicycle-free at radius~$c \log_{d-1} n$;
  \item The total number of cycles of length at most
    $c \log_{d - 1} n$ is $O(n^c)$;
  \item $\lambda(G) \leq 2\sqrt{d-1}\cdot(1+\eps)$.
  \end{itemize}
  
  Finally, by \Cref{thm:al13}, these statements remains true in the
  $(\delta, k)$-wise uniform configuration model,
  $\delta \leq 1/n^{16k+1}$.
\end{theorem}

\end{document}